\def\useieeelayout{0}
\def\showall{0}

\newcommand{\inConf}[1]{\if\showall1{\color{green!50!black}In ACC: #1}\else\if\useieeelayout1{#1}\fi\fi}

\newcommand{\inArxiv}[1]{\if\showall1{\color{blue}In ArXiV: #1}\else\if\useieeelayout0{#1}\fi\fi}

\if\useieeelayout1
\documentclass[letterpaper, 10 pt, conference]{ieeeconf}  
\IEEEoverridecommandlockouts                              %
\usepackage{graphicx}
\usepackage{amsmath}

\usepackage{amsthm}
\usepackage{amssymb}
\usepackage{amsfonts}  
\usepackage{stmaryrd}
\usepackage{multicol}
\usepackage{cuted}
\usepackage{multirow}
\usepackage{caption, subcaption}
\newtheorem{theorem}{Theorem}
\newtheorem{lemma}{Lemma}
\newtheorem{proposition}{Proposition}
\newtheorem{definition}{Definition}
\newtheorem{problem}{Problem}
\newtheorem{assumption}{Assumption}
\newtheorem{remark}{Remark}
\newtheorem{conjecture}[theorem]{Conjecture}
\newtheorem{corollary}{Corollary}
\usepackage{xcolor}
\usepackage{cite}
\usepackage[colorlinks=true,allcolors=steelblue]{hyperref}
\definecolor{steelblue}{RGB}{70,130,180}
\definecolor{mycolor1}{rgb}{0.00000,0.44700,0.74100}%
\definecolor{mycolor2}{rgb}{0.85000,0.32500,0.09800}%
\definecolor{mycolor3}{rgb}{0.92900,0.69400,0.12500}%
\definecolor{mycolor4}{rgb}{0.49400,0.18400,0.55600}%
\definecolor{mycolor5}{rgb}{0.46600,0.67400,0.18800}%
\definecolor{mycolor6}{rgb}{0.30100,0.74500,0.93300}%
\definecolor{mycolor7}{rgb}{0.63500,0.07800,0.18400}%
\usepackage{comment}
\usepackage{algorithm}
\usepackage{algorithmic}
\usepackage{scalerel}
\usepackage{bm}
\usepackage{pgfplots}
\usepackage{tikz}
\usetikzlibrary{plotmarks}
\pgfplotsset{compat=newest}
\pgfplotsset{plot coordinates/math parser=false}
\usepackage{grffile}
\pgfplotsset{compat=newest}
\usetikzlibrary{plotmarks}
\usetikzlibrary{arrows.meta}
\usepgfplotslibrary{patchplots}

\def\bbn{\mathbb N}
\def\bbz{\mathbb Z}
\def\bbr{\mathbb R}

\def\bbp{\mathbb P}

\def\calu{\mathcal U}

\newcommand{\norm}[1]{\left\lVert {#1} \right\rVert}

\title{\LARGE \bf
Regret and Conservatism of Distributionally Robust Constrained Stochastic Model Predictive Control
}

\author{Maik Pfefferkorn, Venkatraman Renganathan, and Rolf Findeisen
\thanks{This project has received funding from the European Research Council (ERC) under the European Union’s Horizon 2020 research and innovation program under grant agreement No 834142 (Scalable Control) and from the German Research Foundation (Research Training Group 2297). M. Pfefferkorn and R. Findeisen are with the Control and Cyber-Physical Systems Laboratory, Technical University of Darmstadt, Germany. M. Pfefferkorn is further with the Institute of Mathematical Optimization, Otto-von-Guericke University Magdeburg, Germany. V. Renganathan is with the Department of Automatic Control - LTH, Lund University, Sweden and is a member of the ELLIIT Strategic Research Area in Lund University. Emails: venkatraman.renganathan@control.lth.se, maik.pfefferkorn@iat.tu-darmstadt.de, rolf.findeisen@iat.tu-darmstadt.de}%
}

\else
\documentclass[10pt]{article}
\usepackage{graphicx}
\usepackage{amsmath}

\usepackage{amsthm}
\usepackage{amssymb}
\usepackage{amsfonts} 
\usepackage{stmaryrd}
\usepackage{multicol}
\usepackage{cuted}
\usepackage{multirow}
\usepackage{caption, subcaption}
\newtheorem{theorem}{Theorem}
\newtheorem{lemma}{Lemma}
\newtheorem{proposition}{Proposition}
\newtheorem{definition}{Definition}
\newtheorem{problem}{Problem}
\newtheorem{assumption}{Assumption}

\newtheorem{corollary}{Corollary}
\usepackage{xcolor}
\usepackage{cite}
\usepackage[colorlinks=true,allcolors=steelblue]{hyperref}
\definecolor{steelblue}{RGB}{70,130,180}
\definecolor{mycolor1}{rgb}{0.00000,0.44700,0.74100}%
\definecolor{mycolor2}{rgb}{0.85000,0.32500,0.09800}%
\definecolor{mycolor3}{rgb}{0.92900,0.69400,0.12500}%
\definecolor{mycolor4}{rgb}{0.49400,0.18400,0.55600}%
\definecolor{mycolor5}{rgb}{0.46600,0.67400,0.18800}%
\definecolor{mycolor6}{rgb}{0.30100,0.74500,0.93300}%
\definecolor{mycolor7}{rgb}{0.63500,0.07800,0.18400}%
\usepackage{comment}
\usepackage{algorithm}
\usepackage{algorithmic}
\usepackage{scalerel}
\usepackage{bm}
\usepackage{url}
\usepackage{pgfplots}
\usepackage{tikz}
\usetikzlibrary{plotmarks}
\pgfplotsset{compat=newest}
\pgfplotsset{plot coordinates/math parser=false}
\usepackage{grffile}
\pgfplotsset{compat=newest}
\usetikzlibrary{plotmarks}
\usetikzlibrary{arrows.meta}
\usepgfplotslibrary{patchplots}

\usepackage[preprint]{tmlr}

\title{Regret and Conservatism of Distributionally Robust Constrained Stochastic Model Predictive Control}

\author{\name Maik Pfefferkorn \email maik.pfefferkorn@iat.tu-darmstadt.de \\
      \addr Control and Cyber-Physical Systems Laboratory \\ Technical University of Darmstadt, Germany \\  \\ \addr Intitute of Mathematical Optimization \\ Otto-von-Guericke University Magdeburg, Germany
      \AND
      \name Venkatraman Renganathan \email venkatraman.renganathan@control.lth.se \\
      \addr Department of Automatic Control - LTH \\ Lund University, Sweden
      \AND
      \name Rolf Findeisen \email rolf.findeisen@iat.tu-darmstadt.de \\
      \addr Control and Cyber-Physical Systems Laboratory \\ Technical University of Darmstadt, Germany
      }

\def\month{09}  
\def\year{2023} 
\def\openreview{\url{https://openreview.net/forum?id=XXXX}} 
\fi

\begin{document}

\maketitle
\thispagestyle{empty}
\pagestyle{empty}

\begin{abstract}

We analyse the conservatism and regret of distributionally robust (DR) stochastic model predictive control (SMPC) when using moment-based ambiguity sets for modeling unknown uncertainties. 
To quantify the conservatism, we compare the deterministic constraint tightening while taking a DR approach against the optimal tightening when the exact distributions of the stochastic uncertainties are known. 
Furthermore, we quantify the sub-optimality gap and regret by comparing the performance when the distributions of the stochastic uncertainties are known and unknown.
Analysing the accumulated sub-optimality of SMPC due to the lack of knowledge about the true distributions of the uncertainties marks the novel contribution of this work.

\end{abstract}

\section{Introduction}\label{sec_intro}

Recently, there has been a surge of interest in analyzing control algorithms that operate subject to unknown \emph{quantities of interest} (QIs) through the lens of regret analysis.
Lack of knowledge about a QI induces regret for the controller, i.e., a performance loss compared to when the QI was known.
Notably, many robust control algorithms incorporate caution mechanisms into their decision making scheme in order to account for the lack of knowledge about the QIs. 
The cautiousness in the decision making can be quantified as the conservatism of the control algorithm which in turn incurs regret. 
For instance, the regret of $H_{\infty}$ control compared against an optimal controller that knows the disturbance sequences exactly beforehand was examined in \cite{karapetyan2022regret}. 
Further works investigated regret in different variants of the finite- and infinite-horizon linear-quadratic regulator problem, see e.g. \cite{Li2019, pmlr-v211-chen23a} and references therein.

In this research, we are interested in comparing controllers that are robust with respect to a moment-based ambiguity set of distributions against a fully informed counterpart that knows the true uncertainty distribution (i.e. QI). 
Specifically, we study the conservatism and regret of the DR SMPC formulation for unknown uncertainty distributions compared against its fully informed counterpart. There have been recent works on DR MPC \cite{marta_smpc, rdm_tac} and more specifically on investigating its regret \cite{yan2023distributionally}. However, they all consider Wasserstein-based formulations. On the contrary, we consider a moment-based ambiguity set formulation. 

Our main contributions are: 1) We define \textsl{constraint conservatism, sub-optimality gap} and \textsl{distributional regret} of DR SMPC with uncertainties modeled using moment-based ambiguity sets, and 2) We develop a framework for quantifying conservatism and for analyzing the resulting regret behavior. We derive analytic expressions to identify and analyze the effects that lead to regret.

The paper is organised as follows: The DR SMPC problem is introduced in Section \ref{sec_problem_formulation}. 
Regret and conservatism associated with DR SMPC are presented in Section \ref{sec_regret}, while some of its features are demonstrated using numerical simulation in Section \ref{sec_num_sim} before concluding in Section \ref{sec_conclusion}. 


\noindent \textbf{Notations:} 
The set of real numbers, integers and the natural numbers are denoted by $\bbr, \bbz, \bbn$ respectively and the subset of real numbers greater than a given constant $a \in \bbr$ is denoted by $\bbr_{> a}$. 
The subset of natural numbers between two constants $a, b \in \bbn$ with $a < b$ is denoted by $[a:b]$. 
For $N \in \mathbb{N}$, we denote by $\llbracket N \rrbracket := \{1,\dots,N \}$.
Given two sets $A \subset \mathbb{R}^{n}, B \subset \mathbb{R}^{n}$, their 
Pontryagin difference is denoted by $A \ominus B := \{ a \in A \mid a + b \in A, \forall b \in B \} \subset \mathbb{R}^{n}$. 
For a matrix $A \in \bbr^{n \times n}$, we denote its transpose and its trace by $A^{\top}$ and $\mathbf{Tr}(A)$ respectively. 
Given $x \in \bbr^{n}$, and $A \in \bbr^{n \times n}$, the notation ${\left \| x \right \Vert}^{2}_{A}$ denotes $x^{\top} A x$. 

\section{DR SMPC Problem Formulation}
\label{sec_problem_formulation}

We consider DR SMPC problem with joint chance constraints similar to \cite{paulson2020stochastic} using moment-based ambiguity set modeling to describe the stochastic system uncertainties.

\subsection{System Dynamics \& Constraints}
We consider a stochastic linear time-invariant system   
\begin{equation} \label{eqn_system_dynamics}
x_{k+1}= Ax_k + Bu_k + w_k, \quad \forall k \in \mathbb{N},
\end{equation}
where $x_{k} \in \bbr^n$ and $u_{k} \in \bbr^m$ are the system state and input at time step $k$, respectively. The matrices $A \in \mathbb{R}^{n \times n}$ and $B \in \mathbb{R}^{n \times m}$ denote the dynamics matrix and the input matrix respectively. For the ease of exposition, we will assume that $x_{0} \in \mathbb{R}^{n}$ is known. The process noise $w_k \in \bbr^{n}$ is a zero-mean random vector that is independent and identically distributed over time. 
The distribution of $w_k$, namely $\mathbb{P}^{w}$, is unknown.
However, it is known to belong to a moment-based ambiguity set of distributions given by 
\begin{equation} \label{eqn_w_ambig_set}
\mathcal{P}^w = \left\{ \mathbb{P}^{w} \mid \mathbb{E}[w_k]=0, \mathbb{E}[w_k w_k^{\top}] = {\Sigma}_{w} \right\}.
\end{equation}
We denote by $N \in \bbn$ the prediction horizon of the predictive control problem to be defined shortly.
The states, inputs and the disturbances over the prediction horizon are given by
\begin{align}
\mathbf{x}_k &:= \begin{bmatrix} x^{\top}_{0 \mid k} & x^{\top}_{1 \mid k} & \cdots & x^{\top}_{N \mid k} \end{bmatrix}^{\top} \in \mathbb{R}^{(N+1)n}, \\
\mathbf{u}_k &:= \begin{bmatrix} u^{\top}_{0 \mid k} & u^{\top}_{1 \mid k} & \cdots & u^{\top}_{N-1 \mid k} \end{bmatrix}^{\top}  \in \mathbb{R}^{Nm},
\end{align}
\begin{align}
\mathbf{w}_k &:= \begin{bmatrix} w^{\top}_{0 \mid k} & w^{\top}_{1 \mid k} & \cdots & w^{\top}_{N-1 \mid k} \end{bmatrix}^{\top}  \in \mathbb{R}^{Nn}.
\end{align}
Here, the subscript $i \! \mid \! k$ indicates $i$ time steps ahead of $k$ and $x_{0 \mid k} = x_k$.
Starting at $x_{k}$, we compactly write the evolution of \eqref{eqn_system_dynamics} over the prediction horizon as
\begin{align} \label{eqn_compact_system_model}
    \mathbf{x}_k = \mathbf{A} x_{0 \mid k} + \mathbf{B} \mathbf{u}_k + \mathbf{D} \mathbf{w}_k,
\end{align}
with matrices $\mathbf{A}, \mathbf{B}$, and $\mathbf{D}$ being of appropriate dimension. 
Note that $\mathbf{x}_k$ and $\mathbf{w}_k$ are random vectors as the realizations of future disturbances $w_{i \mid k}, i \in [0:N-1]$ are a-priori unknown.
The mean and covariance of \eqref{eqn_compact_system_model} evolve as
\begin{align}
    \bar{\mathbf{x}}_k &:= \mathbb{E}[\mathbf{x}_k] =  \mathbf{A} x_{0 \mid k} + \mathbf{B} \mathbf{u}_k, \label{eqn_mean_dynamics} \\
    \mathbf{\Sigma}_{\mathbf{x}} &:= \mathbb{E}[(\mathbf{x}_k - \bar{\mathbf{x}}_k)(\mathbf{x}_k - \bar{\mathbf{x}}_k)^{\top}] = \mathbf{D} \mathbf{\Sigma}_{\mathbf{w}} \mathbf{D}^{\top}, \label{eqn_covariance_dynamics}
\end{align}
where $\mathbf{\Sigma}_{\mathbf{w}}$ is a block diagonal matrix with $N$ blocks and each block being equal to $\Sigma_{w}$.
Note that the state and disturbance sequences $(x_i)_{i \in [0:k]}$ and $(w_i)_{i \in [0:k-1]}$ have been realized at time step $k$ and are hence deterministic.
We denote those realizations by $(\hat{x}_i)_{i \in [0:k]}$ and $(\hat{w}_i)_{i \in [0:k-1]}$, respectively, where $\hat{x}_0 = x_0$ is the deterministic initial condition.
Since the support of $\mathbb{P}^{w}$ can be unbounded, we cannot guarantee the satisfaction of hard state constraints for all $k$. Hence, we enforce a DR joint risk constraint on the states as 
\begin{align} \label{eqn_risk_constraints}
    \sup_{\bbp_k^{\mathbf{x}} \in \mathcal{P}^{\mathbf{x}}_k } \bbp_k^{\mathbf{x}} \left[ \mathbf{x}_k \notin \mathcal{X} \right] \leq \Delta,
\end{align}
where $\mathcal{P}^{\mathbf{x}}_k$ denotes the moment-based ambiguity set for the compact state $\mathbf{x}_k$ and is given by
\begin{align*}
    \mathcal{P}^{\mathbf{x}}_k = \left\{ \mathbb{P}^{\mathbf{x}}_k \mid \mathbb{E}[\mathbf{x}_k] = \bar{\mathbf{x}}_k, \mathbb{E}[(\mathbf{x}_k \! - \! \bar{\mathbf{x}}_k) (\mathbf{x}_k - \bar{\mathbf{x}}_k)^{\top}] = \mathbf{\Sigma}_{\mathbf{x}} \right\}.
\end{align*}
The set $\mathcal{X}$ is assumed to be a convex polytope defined by  
\begin{align}
\mathcal{X} &:= \bigcap^{n_{x}}_{i = 1} \left\{ \mathbf{x} \mid f^{\top}_{i} \mathbf{x} \leq g_{i} \right\} = \left\{ \mathbf{x} \mid \mathbf{F} \mathbf{x} \leq \mathbf{g} \right\},   
\end{align}
where $f_{i} \in \mathbb{R}^{(N+1)n}, g_{i} \in \mathbb{R}, \mathbf{F} \in \mathbb{R}^{n_{x} \times (N+1)n}$ and $\mathbf{g} \in \mathbb{R}^{n_{x}}$. 
Further, $\Delta \in (0,0.5]$ represents a user-prescribed total risk budget for the worst-case probability of constraint violation over the entire prediction horizon. 
A similar probabilistic treatment of control constraints can also be formulated. For the ease of exposition, we consider the hard input constraint formulation $u_k \in \mathcal{U}, \forall k$, where $\mathcal{U}$ is a convex polytope
\begin{align}
\mathcal{U} &:= \bigcap^{n_{u}}_{j = 1} \left\{ \mathbf{u} \mid c^{\top}_{j} \mathbf{u} \leq d_{j} \right\} = \left\{ \mathbf{u} \mid \mathbf{C} \mathbf{u} \leq \mathbf{d} \right\} \label{eqn_controlconstraint_polytope}    
\end{align}
with $c_{j} \in \mathbb{R}^{Nm}, d_{i} \in \mathbb{R}, \mathbf{C} \in \mathbb{R}^{n_{u} \times Nm}$ and $\mathbf{d} \in \mathbb{R}^{n_{u}}$.
We employ SMPC as a solution approach to minimizing the cost
\begin{equation}
    J_\infty(\mathbf{x}_\infty, \mathbf{u}_\infty) = \sum_{k=0}^\infty \lVert x_k \rVert_Q^2 + \lVert u_k \rVert_R^2,
\end{equation}
where $\mathbf{x}_\infty = (x_i)_{i \in \bbn}$, $\mathbf{u}_\infty = (u_i)_{i \in \bbn}$, $Q \in \bbr^{n \times n}, Q \succeq 0$ and $R \in \bbn^{m \times m}, R \succ 0$ and which is an intractable problem due to the infinite horizon as well as the stochastic dynamics.
The objective of the SMPC is to minimize the cost
\begin{align}\label{eq:smpc_cost_orig}
J_{\mathrm{SMPC}}(\mathbf{u}_k, x_{0 \mid k}) := \mathbb{E} \left[ \norm{\mathbf{x}_k}^{2}_{\mathbf{Q}} + \norm{\mathbf{u}_k}^{2}_{\mathbf{R}}\right]    
\end{align}
defined for a finite horizon $N$ and where the expectation $\mathbb{E}[ \cdot ]$ accounts for the stochasticity of the states.
Note that $\mathbf{Q}$ is a block-diagonal matrix with $N+1$ blocks of which the first $N$ are equal to $Q \succeq 0$ and the last block is given by $Q_\mathrm{f} \succeq 0$. 
For appropriately designed $Q_\mathrm{f}$, \eqref{eq:smpc_cost_orig} approximates the expected, remaining infinite-horizon cost starting at $x_k$ at time point $k$.
Similarly, $\mathbf{R}$ is a block-diagonal matrix with $N$ blocks and each block being equal to $R \succ 0$.
Then, $J_{\mathrm{SMPC}}(\mathbf{u}_k, x_{0 \mid k})$ can be rewritten as
\begin{align} \label{eqn_cost_reform1}
    J_{\mathrm{SMPC}}(\mathbf{u}_k, x_{0 \mid k}) = \norm{\bar{\mathbf{x}}_k}^{2}_{\mathbf{Q}} + \norm{\mathbf{u}_k}^{2}_{\mathbf{R}} + \mathbf{Tr}(\mathbf{Q} \mathbf{\Sigma}_{\mathbf{x}}).
\end{align}
We now formally state the DR SMPC optimization problem along with the state and input constraints. 

\subsection{The DR SMPC Problem}

\begin{problem} \label{smpc_problem}
Given an initial state $x_{0 | k} = x_k \in \mathbb{R}^{n}$, risk budget $\Delta$, process noise $w_{i \mid k} \sim \mathbb{P}^{w} \in \mathcal{P}^{w}, \forall i \in [0 : N-1]$, and the penalty matrices $Q \succeq 0, R \succ 0$, we seek to find an input sequence $\mathbf{u}_k$ that optimizes the following optimal control problem (OCP): 
\begin{subequations} \label{eqn_smpc_optimization_problem}
\begin{alignat}{2}
&\! \underset{\mathbf{u}_k}{\mathrm{min}} &\qquad& J_{\mathrm{SMPC}}(\mathbf{u}_k, x_{0 \mid k}) \label{eq:optProb}\\
&\mathrm{s.} \, \mathrm{t.} &      & \mathbf{x}_k = \mathbf{A} x_{0 \mid k} + \mathbf{B} \mathbf{u}_k + \mathbf{D} \mathbf{w}_k, \quad x_{0 \mid k} = x_k, \label{eq:constraint1}\\
&                  &      & \mathbf{u}_k \in \calu, w_{i \mid k} \sim \mathbb{P}^{w} \in \mathcal{P}^w, \label{eq:constraint2} \\
&                  &      & \sup_{\bbp_k^{\mathbf{x}} \in \mathcal{P}^{\mathbf{x}}_k } \bbp_k^{\mathbf{x}} \left[ \mathbf{x}_k \notin \mathcal{X} \right] \leq \Delta. \label{cc_joint}
\end{alignat}
\end{subequations}
\end{problem}
Problem \ref{smpc_problem} is solved online in a receding horizon fashion at each time instance $k$.
The first element $u^{\star}_{0 \mid k}$ of the optimal input sequence $\mathbf{u}^{\star}_{k}$ is applied to the system given by \eqref{eqn_system_dynamics} until the next sampling instant $k+1$\footnote{We refer the reader to \cite{Rawlings2019} for a comprehensive introduction to and overview of model predictive control.}. 
Note that the OCP \eqref{eqn_smpc_optimization_problem} is in general computationally intractable as \eqref{cc_joint} is an infinite-dimensional DR joint risk constraint. Using Boole's inequality, \eqref{cc_joint} can be decomposed into individual risk constraints across each time step along the prediction horizon. Suppose that $\sum^{n_{x}}_{i=1} \delta_{i} \leq \Delta$ and
\begin{align} 
 \sup_{\bbp_k^{\mathbf{x}} \in \mathcal{P}^{\mathbf{x}}_k } \bbp_k^{\mathbf{x}} \left[ f^{\top}_{i} \mathbf{x}_k > g_{i} \right] &\leq \delta_{i}, \quad \forall i \in \llbracket n_{x} \rrbracket. 
 \label{eqn_individual_dr_constraint}
\end{align}
Then, \eqref{eqn_individual_dr_constraint} implies \eqref{cc_joint}.%
\inArxiv{
That is, 
\begin{align*}
\sup_{\bbp_k^{\mathbf{x}} \in \mathcal{P}^{\mathbf{x}}_k } \bbp_k^{\mathbf{x}} \! \left[ \mathbf{x}_k \! \notin \! \mathcal{X} \right] 
&= \sup_{\bbp_k^{\mathbf{x}} \in \mathcal{P}^{\mathbf{x}}_k } \bbp_k^{\mathbf{x}} \left[ \mathbf{x}_k \! \in \! \bigcup^{n_{x}}_{i = 1} \left\{ \mathbf{x}_k \! \mid \! f^{\top}_{i} \mathbf{x}_k \! > \! g_{i} \right\} \right]   \\
&\leq \sum^{n_{x}}_{i = 1} \sup_{\bbp_k^{\mathbf{x}} \in \mathcal{P}^{\mathbf{x}}_k } \bbp_k^{\mathbf{x}} \left[ f^{\top}_{i} \mathbf{x}_k > g_{i} \right] \\
&\leq \sum^{n_{x}}_{i=1} \delta_{i} \\
&\leq \Delta.
\end{align*}
}%
Note that the decomposition using Boole's inequality is known to be conservative, see \cite{hunter_book} for less conservative alternatives.
Though \eqref{eqn_individual_dr_constraint} is an infinite-dimensional risk constraint, it can be equivalently formulated as a second order cone constraint on $\bar{\mathbf{x}}_{k}$ through deterministic constraint tightening (See Lemma \ref{lemma_tightening} given without proof).

\begin{lemma} \label{lemma_tightening}
(From \cite{ono2008iterative} and \cite{ElGhaoui_DR_LP}) Let $\mathbf{x} \in \mathbb{R}^{n}$ be a random vector with known mean $\bar{\mathbf{x}}$ and covariance $\mathbf{\Sigma}_{\mathbf{x}}$. Then, $\forall \delta_{i} \! \in \! (0, 0.5], i \in \llbracket n_x \rrbracket$, we have
\begin{equation}
\footnotesize
\bbp^{\mathbf{x}} \left[ f^{\top}_{i} \mathbf{x} > g_{i} \right] 
\leq \delta_{i} \nonumber \Leftrightarrow f^{\top}_{i} \bar{\mathbf{x}} \leq g_{i} - \psi_{i} \norm{\mathbf{\Sigma}_{\mathbf{x}}^{\frac{1}{2}}f_{i}}_{2},
\end{equation}
where the tightening constant $\psi_{i}, \forall i \in \llbracket n_{x} \rrbracket$ is given by
\begin{align} \label{eqn_tightening_constant}
    \psi_{i} 
    := \begin{cases}
    \Phi^{-1}_{\mathbb{P}^{\mathbf{x}}}(1-\delta_{i}), &\text{ when $\mathbb{P}^{\mathbf{x}}$ is known}, \\
    \sqrt{ \frac{1-\delta_{i}}{\delta_{i}}}, &\text{ when $\mathbb{P}^{\mathbf{x}}$ is unknown}, 
    \end{cases}
\end{align}
and $\Phi_{\mathbb{P}^{\mathbf{x}}}$ denotes the Cumulative Distribution Function of the known distribution $(\mathbb{P}^{\mathbf{x}})$ when normalized.
\end{lemma}

The surrogate of problem \ref{smpc_problem}, called the surrogate SMPC, can then be written as
\begin{subequations} \label{eqn_smpc_optimization_problem_surrogate}
\begin{alignat}{2}
&\! \underset{\mathbf{u}_k}{\mathrm{min}} &\qquad& J_{\mathrm{SMPC}}(\mathbf{u}_k, x_{0 \mid k}) \label{eq:optProb_surrogate}\\
&\mathrm{s.} \, \mathrm{t.} &      & \mathbf{x}_k = \mathbf{A} x_{0 \mid k} + \mathbf{B} \mathbf{u}_k + \mathbf{D} \mathbf{w}_k, \quad x_{0 \mid k} = x_k, \label{eq:constraint1_surrogate}\\
&                  &      & \mathbf{u}_k \in \calu, w_{k} \sim \mathbb{P}^{w} \in \mathcal{P}^w, \label{eq:constraint2_surrogate} \\
&                  &      & f^{\top}_{i} \bar{\mathbf{x}}_k \leq g_{i} - \psi_{i} \norm{\mathbf{\Sigma}_{\mathbf{x}}^{\frac{1}{2}}f_{i}}_{2}, i \in \llbracket n_{x} \rrbracket.
\label{cc_joint_surrogate} 
\end{alignat}
\end{subequations}
Note that \eqref{eqn_smpc_optimization_problem_surrogate} has finite-dimensional constraints unlike its original counterpart in \eqref{eqn_smpc_optimization_problem}. 
Let the control input sequences that minimize \eqref{eqn_smpc_optimization_problem_surrogate} with exact and 
DR constraint tightening according to \eqref{eqn_tightening_constant}
be denoted by $\mathbf{u}^{\star}_{k}$ and $\mathbf{u}^{\dagger}_{k}$ respectively.
Their corresponding optimal value functions are given by
\begin{align}\label{eq:optimal_value_functions}
    J^{\diamond}_{\mathrm{SMPC}}(x_k) &:= J_{\mathrm{SMPC}}(\mathbf{u}^{\diamond}_{k}, x_k), \quad \text{where~} \diamond \in \{\dagger, \star \}.
\end{align}

\inConf{
\begin{definition}(From \cite{Culbertson2023})\label{def:ISSp}
    System \eqref{eqn_system_dynamics} is input-to-state stable in probability (ISSp) with respect to $L^p$ if $ \forall \epsilon \! \in \! (0,1)$, $K \! \in \! \bbn$ and $w_k \! \in \! L^p$, $\exists \beta \! \in \! \mathcal{KL}$ and $\exists \varrho \! \in \! \mathcal{K}$ such that\footnote{See appendix of \cite{Pfefferkorn2023} for the definitions of $L^p$ spaces for random vectors and function classes $\mathcal{KL}$ and $\mathcal{K}$.}
    \begin{equation*}
        \mathbb{P}_x \left [ \lVert x_{k+i} \rVert \leq \beta( \lVert x_{k} \rVert, i) \! + \! \varrho( \lVert w_{k+i} \rVert_{L^p}), \forall i \! \leq \! K \right ] \! \geq \! 1 \! - \! \epsilon.
    \end{equation*}
    If this holds for $\beta( \lVert x_k \rVert, i) = M \nu^{i} \lVert x_k \rVert$ with $M>0$ and $\nu \in (0,1)$, system \eqref{eqn_system_dynamics} is exponentially ISSp (eISSp).
\end{definition}}

\inArxiv{
\begin{definition}(From \cite{Culbertson2023})\label{def:ISSp}
    System \eqref{eqn_system_dynamics} is input-to-state stable in probability (ISSp) with respect to $L^p$ if $ \forall \epsilon \! \in \! (0,1)$, $K \! \in \! \bbn$ and $w_k \! \in \! L^p$, $\exists \beta \! \in \! \mathcal{KL}$ and $\exists \varrho \! \in \! \mathcal{K}$ such that\footnote{See appendix for the definitions of $L^p$ spaces for random vectors and function classes $\mathcal{KL}$ and $\mathcal{K}$.}
    \begin{equation*}
        \mathbb{P}_x \left [ \lVert x_{k+i} \rVert \leq \beta( \lVert x_{k} \rVert, i) \! + \! \varrho( \lVert w_{k+i} \rVert_{L^p}), \forall i \! \leq \! K \right ] \! \geq \! 1 \! - \! \epsilon.
    \end{equation*}
    If this holds for $\beta( \lVert x_k \rVert, i) = M \nu^{i} \lVert x_k \rVert$ with $M>0$ and $\nu \in (0,1)$, system \eqref{eqn_system_dynamics} is exponentially ISSp (eISSp).
\end{definition}}

Similar to ISS in the case of bounded disturbances, the ISSp property can be formulated in terms of a Lyapunov function as follows.

\begin{definition}(From \cite{Culbertson2023})
    A continuous function $V: \bbr^n \rightarrow \bbr_{\geq 0}$ is an ISSp Lyapunov function for system \eqref{eqn_system_dynamics} if there exist functions $\kappa_1, \kappa_2, \kappa_3 \! \in \! \mathcal{K}_\infty$ and $\kappa_4 \! \in \! \mathcal{K}$ such that
    \begin{subequations}
    \begin{align*}
            & \kappa_1( \lVert x_k \rVert ) \leq V(x_k) \leq \kappa_2( \lVert x_k \rVert ) \\
            & \mathbb{E}[ V(f(x_k, u_k) \! - \! V(x_k) ] \leq -\kappa_3( V( x_k ) ) \! + \! \kappa_4( \lVert w_k \rVert_{L^p})
    \end{align*}
    \end{subequations}
    hold for all $x_k \in \mathcal{X}$ and $w_k \in L^p$.
    If $\kappa_1(\lVert x_k \rVert) = a \lVert x_k \rVert^c$, $\kappa_2(\lVert x_k \rVert) = b \lVert x_k \rVert^c$ and $\kappa_3(V(x)) = \nu V(x)$ with $a,b,c \in \bbr_{>0}$ and $\nu \in (0,1)$, then $V$ is an eISSp Lyapunov function.
\end{definition}

\begin{assumption}\label{as:ISSp}
    System \eqref{eqn_system_dynamics} is exponentially input-to-state stable in probability (eISSp) under DR SMPC for both the fully informed and the DR cases. That is, $\exists \beta^{\star}, \beta^\dagger \in \mathcal{KL}$, given by $\beta^{\star}(s, t) = M^{\star} \nu^{\star^t} s, \beta^\dagger(s, t) = M^\dagger \nu^{\dagger^t} s$ with $M^{\star},M^\dagger>0$ and $\nu^{\star},\nu^\dagger \in (0,1)$ such that $\forall i \in \llbracket K^\diamond \rrbracket, K^\diamond \in \bbn$
    \begin{equation}\label{eq:eISSp_by_assumption}
        \mathbb{P}_x \left [ \lVert x_{k+i}^\diamond \rVert \leq \beta^\diamond( \lVert x_{k}^\diamond \rVert, i) \! + \! \varrho( \lVert w_{k+i} \rVert_{L^p}) \right ] \! \geq \! 1 \! - \! \varepsilon
    \end{equation}
    \color{black}
    holds for $w_k \in L^p, w_k \sim \mathbb{P}^w$ for some $p > 0$ and $\varrho \in \mathcal{K}$, given $\varepsilon \! \in \! (0,1)$. The optimal value functions $J_{\mathrm{SMPC}}^{\star}(x_k), J_{\mathrm{SMPC}}^{\dagger}(x_k)$ are eISSp Lyapunov functions, i.e., $\mathbb{E}[J^\diamond_\mathrm{SMPC}(x_{k+1}) - J^\diamond_\mathrm{SMPC}(x_k) \mid x_k] \leq - \omega^\diamond J_\mathrm{SMPC}^\diamond(x_k) + \rho(\|w_k\|_{L^p})$ for $\omega^\diamond \in (0,1)$ and $\rho \in \mathcal{K}$.
\end{assumption}

\begin{corollary}\label{corollary:nominal_stability}(From \cite{Culbertson2023})
    The nominal (undisturbed) system is asymptotically stable according to
    \begin{equation}\label{eq:nominal_stability}
        \mathbb{P}[ \lVert \bar{x}_{k+i} \rVert \leq \beta(\lVert x_k \rVert, i), \forall i \leq K)] = 1.
    \end{equation}
\end{corollary}

Note that there is an interplay between the probability $\varepsilon$ and the horizon lengths $K^{\star}, K^\dagger$.
For fixed $\varepsilon$, the (maximum) horizon lengths $K^{\star}, K^\dagger$ for which eISSP can be established in the respective case is determined by the problem formulation and vice versa.
Furthermore, lower values of $\varepsilon$ imply lower values of $K^{\star}, K^\dagger$, in general.
We consider the common horizon length $\bar{K} = \min \{ K^{\star}, K^\dagger \}$ in the following for a valid analysis in both the DR and the fully informed case.
Note that stability of MPC is usually achieved by design, which is, however, out of the scope of this work.


\section{Conservatism \& Regret Analyses} \label{sec_regret}
In this section, we define the concepts of conservatism and regret for the previously introduced DR SMPC algorithm. 

\subsection{Conservatism of DR SMPC}
We would like to study the difference in constraint tightening when the true distributions of the stochastic uncertainties are known and when they are unknown. 

\begin{definition}
The constraint conservatism, denoted by $\mathfrak{C} \in \mathbb{R}$, associated with the DR SMPC is defined as the difference in volume of the deterministically tightened state constraint set 
with and without the knowledge of $\mathbb{P}^{w}$ and $\mathbb{P}_k^{\mathbf{x}}$ respectively. That is,
\begin{align}
\mathfrak{C} 
&:= 
\int_{\underbrace{\left(\mathcal{X}_{\mathrm{True}} \ominus \mathcal{X}_{\mathrm{DR}} \right)}_{:= \mathcal{X}_{\mathrm{Diff}}}} dx, \quad \text{where}, \label{eqn_constraint_conservatism} \\
\mathcal{X}_{\mathrm{True}} 
&= 
\left\{ \mathbf{x} \mid \mathbf{F} {\mathbf{x}} \leq \mathbf{g}_{\mathrm{True}} \right\}, \\
\mathcal{X}_{\mathrm{DR}} 
&= 
\left\{ \mathbf{x} \mid \mathbf{F} {\mathbf{x}} \leq \mathbf{g}_{\mathrm{DR}} \right\}, \quad \text{where} \quad \forall i \in \llbracket n_{x} \rrbracket, \\
\begin{bmatrix} \mathbf{g}_{\mathrm{True}} \end{bmatrix}_i
&:= 
\begin{bmatrix}
g_{i} - \Phi^{-1}_{\mathbb{P}_k^{\mathbf{x}}}(1-\delta_{i}) \norm{\mathbf{\Sigma}_{\mathbf{x}}^{\frac{1}{2}}f_{i}}_{2}
\end{bmatrix}, \quad \text{and}, \label{eqn_true_constraint_tightening} \\
\begin{bmatrix} \mathbf{g}_{\mathrm{DR}} \end{bmatrix}_i
&:= 
\begin{bmatrix}
g_{i} - \sqrt{\frac{1 - \delta_{i}}{\delta_{i}}} \norm{\mathbf{\Sigma}_{\mathbf{x}}^{\frac{1}{2}}f_{i}}_{2} 
\end{bmatrix}. \label{eqn_DR_constraint_tightening}
\end{align}

\end{definition}
Let $\mathbf{z} := \left[
\begin{array}{c|c}
\mathbf{x}^\top & \mathbf{x}^\top 
\end{array}\right]^\top \in \mathbb{R}^{2n(N+1)}$ and
\begin{align}
\mathbf{h} := \left[
\begin{array}{c}
\mathbf{g}_{\mathrm{True}} \\ \hline \mathbf{g}_{\mathrm{DR}} 
\end{array}\right] \in \mathbb{R}^{2 n_{x}}, \, \mathbf{H} := \left[
\begin{array}{c|c}
\mathbf{F}  & \mathbf{F} \\ \hline
\mathbf{0}_{n_{x} \times n(N+1)} & \mathbf{F}
\end{array}\right].     
\end{align}
Then, using Theorem 3.3 in \cite{gabidullina2019minkowski}, we can express $\mathcal{X}_{\mathrm{Diff}}$ as
\begin{align}\label{eq:X_Diff}  
\mathcal{X}_{\mathrm{Diff}} 
:= \left\{ \mathbf{z} \mid \mathbf{H} {\mathbf{z}} \leq \mathbf{h} \right\}. 
\end{align}
Although analyzing \eqref{eq:X_Diff} is in general hard, its volume (i.e., conservatism) can be computed numerically, see e.g. \cite{Chevallier2022}. 

\textbf{Remark 1:} 
\inArxiv{
Note that \eqref{eqn_constraint_conservatism} uses the Pontryagin difference of the two deterministically tightened convex polyhedrons $\mathcal{X}_{\mathrm{True}}$ and $\mathcal{X}_{\mathrm{DR}}$.
}
The constraint tightening in both \eqref{eqn_true_constraint_tightening} and \eqref{eqn_DR_constraint_tightening} depends upon the risk\footnote{Both $\mathcal{X}_{\mathrm{Diff}}$ and $\mathfrak{C}$ will have interesting observations with a non-uniform risk allocation as in \cite{ono2008iterative}. Analysing this aspect is left for future.} $\delta_{i}$ corresponding to the $i$\textsuperscript{th} constraint.  
For $\delta_{i}, i \in \llbracket n_{x} \rrbracket$ close to $0$, the DR constraint tightening will be significantly stricter than the exact tightening. 

\subsection{Regret of DR SMPC}
We start by noting that for DR SMPC formulation \eqref{eqn_smpc_optimization_problem_surrogate}, regret is introduced by  conservatism: 
Since the deterministic constraint tightening will be different for the case when the true distributions $\mathbb{P}_{w}$ and $\mathbb{P}_{\mathbf{x}}$ are known and unknown respectively, the resulting open-loop costs \eqref{eq:optimal_value_functions} and closed-loop costs till a time point $k$
will be different.
This difference of the optimal costs is essentially referred to as regret. 

\begin{assumption}\label{as:same_disturbance}
    Both the systems under the DR and fully informed SMPC respectively encounter the same disturbance realization $w_k$ at all time steps $k \in \mathbb{N}$.
\end{assumption}
%
%
\begin{definition}
    Given Assumption \ref{as:same_disturbance}, the closed-loop regret accumulated up to time point $k \in \bbn$ is defined as the difference in the closed-loop costs with and without knowledge of $\mathbb{P}^{w}$ and $\mathbb{P}_k^{\mathbf{x}}$ respectively. That is,
    \begin{align}\label{eqn_regretDef_cl}
        \mathfrak{R}_k 
        \! = \! 
        \sum_{i=0}^{k} 
       \left( \lVert \hat{x}_i^\dagger \rVert_Q^2 \! + \! \lVert u^\dagger_i \rVert_R^2 \right) \! - \! \left(\lVert \hat{x}_i^{\star} \rVert_Q^2 \! + \! \lVert u_i^{\star} \rVert_R^2 \right) 
    \end{align}
    where $\cdot^{\star}$ and $\cdot^\dagger$ indicate the fully informed and DR quantities respectively, with $x_0^\dagger = x_0^{\star} = x_0$ and $\mathfrak{R}_0^\mathrm{cl} = 0$.
\end{definition}

Besides $\mathfrak{R}_k$, we can also exploit the information obtained from solving the surrogate SMPC \eqref{eqn_smpc_optimization_problem_surrogate} to define the expected remaining infinite-horizon regret from time point $k$ on in open-loop -- given that \eqref{eqn_cost_reform1} is designed to approximate \eqref{eq:smpc_cost_orig}.

\begin{definition}
    The suboptimality gap, which corresponds to the expected remaining infinite-horizon cost difference in open-loop from time $k \in \bbn$ on, is defined as the difference in the optimal value functions of DR SMPC \eqref{eqn_smpc_optimization_problem_surrogate} with and without knowledge of $\mathbb{P}^{w}$ and $\mathbb{P}_k^{\mathbf{x}}$ respectively. That is,
    \begin{align} \label{eqn_regretDef_ol}
        \mathfrak{G}_k := J^{\dagger}_{\mathrm{SMPC}}(x_k^\dagger) - J^{\star}_{\mathrm{SMPC}}(x_k^{\star}).    
\end{align}
\end{definition}

To analyze $\mathfrak{R}_k$ and $\mathfrak{G}_k$, we derive a closed-form expression for $\mathbf{u}^{\star}_{k}$ and $\mathbf{u}^{\dagger}_{k}$, whose first elements also constitute the closed-loop inputs $u^{\star}_k$ and $u_k^\dagger$.
\color{black}
We reformulate OCP \eqref{eqn_smpc_optimization_problem_surrogate} as a quadratic program (QP) of the form
\begin{subequations} \label{eqn_qp}
\begin{alignat}{2}
    &\! \underset{\mathbf{u_k}}{\mathrm{min}} &\qquad& \frac{1}{2} \norm{\mathbf{u}_k}^{2}_{H} + h_k^\top \mathbf{u}_k + r_k \label{eq:qpCost}\\
    &\mathrm{s.} \, \mathrm{t.} &      & \mathbf{M} \mathbf{u}_k - \mathbf{b}_k \leq \mathbf{0}. \label{eq:qpConstraints}
\end{alignat}
\end{subequations}
We start by reformulating the cost function \eqref{eqn_cost_reform1} of OCP \eqref{eqn_smpc_optimization_problem_surrogate}.
Substituting \eqref{eqn_mean_dynamics} and \eqref{eqn_covariance_dynamics} into \eqref{eqn_cost_reform1}, we obtain 
\begin{align} \label{eqn_cost_reform2}
    J_{\mathrm{SMPC}}(\mathbf{u}_k, x_k) = \frac{1}{2} \norm{\mathbf{u}_k}^{2}_{H} + h_k^{\top} \mathbf{u}_k + r_k, 
\end{align}
where $H = 2(\mathbf{B}^{\top} \mathbf{Q} \mathbf{B} + \mathbf{R}) \succ 0, h_k^{\top} = 2 x^{\top}_{0 \mid k} \mathbf{A}^{\top} \mathbf{Q} \mathbf{B}$ and $r_k = \mathbf{Tr}\left(\mathbf{Q} \mathbf{D} \mathbf{\Sigma}_{w} \mathbf{D}^{\top} \right) + \norm{x_{0 \mid k}}^{2}_{\mathbf{A}^{\top} \mathbf{Q} \mathbf{A}}$. 
Next, we substitute \eqref{eqn_mean_dynamics} and \eqref{eqn_covariance_dynamics} into \eqref{cc_joint_surrogate} to reformulate the tightened state constraints in terms of only the input as 
\begin{align}\label{eq:qp_state_constraints}
    \underbrace{f_i^\top \mathbf{B}}_{=: \bar{f}_i^\top} \mathbf{u}_k \leq \underbrace{g_i - f^\top_i \mathbf{A} x_{0 \mid k}}_{=: \bar{g}_{k,i}} - \psi_i \underbrace{\norm{(\mathbf{D} \mathbf{\Sigma}_{\mathbf{w}} \mathbf{D}^\top)^{\frac{1}{2}} f_i}_2}_{=: v_i}.
\end{align}
Then, $\forall i \in \llbracket n_{x} \rrbracket$, writing \eqref{eq:qp_state_constraints} equivalently in vectorized form as $\mathbf{\bar{F}} \mathbf{u}_k \leq \mathbf{\bar{g}}_k - \bm{\psi}^\top \mathbf{v}$ enables us to define \eqref{eq:qpConstraints} via
\begin{align}\label{eqn_constr_reform}
    \mathbf{M} = \begin{bmatrix} \mathbf{C} \\ \mathbf{\bar{F}} \end{bmatrix}, \quad \mathbf{b}_k = \begin{bmatrix} \mathbf{d} \\ \mathbf{\bar{g}}_k - \mathrm{diag}({\bm{\psi}}) \mathbf{v} \end{bmatrix}.
\end{align}
We now define when an inequality constraint is called active.

\begin{definition}
    An inequality constraint is said to be active if $\mathbf{M}_{i:} \mathbf{u}_k - \mathbf{b}_{k,i} = 0$ and inactive if $\mathbf{M}_{i:} \mathbf{u}_k - \mathbf{b}_{k,i} < 0$, where $\mathbf{M}_{i:}$ denotes the i$^\text{th}$ row of $\mathbf{M}$ and $\mathbf{b}_{k,i}$ the i$^\text{th}$ entry of $\mathbf{b}_k$. The active set $\mathcal{A}_k \subseteq \llbracket n_x+n_u \rrbracket$ is the index set of active inequality constraints.
\end{definition}

\begin{assumption} \label{as:index_set}
    The active set $\mathcal{A}_k^\diamond$ of QP \eqref{eqn_qp} is known $\forall k$.
\end{assumption}

To solve the QP, we assume the following regularity condition on the constraints.

\begin{assumption} \label{as:licq}
    QP \eqref{eqn_qp} satisfies the linear independence constraint qualification (LICQ) criterion, i.e., the gradients of the active inequality constraints are linearly independent.
\end{assumption}

\begin{proposition}\label{prop:qp_solution}
    Under Assumptions \ref{as:index_set} and \ref{as:licq}, the unique and global solution to QP \eqref{eqn_qp} at time step $k$ is given by
    \begin{equation} \label{eq:qp_solution}
    \mathbf{u}_k^\diamond = \left( V_k^\diamond \mathbf{\tilde{M}}^\diamond_k H^{-1} - H^{-1} \right) h_k
    + V_k^\diamond \mathbf{\Tilde{b}}_k^\diamond,
    \end{equation}
    where $V_k^\diamond =  H^{-1} \mathbf{\Tilde{M}}_k^{\diamond^\top} \left( \mathbf{\Tilde{M}}_k^\diamond H^{-1} \mathbf{\Tilde{M}}_k^{\diamond^\top} \right)^{-1}$, $\mathbf{\Tilde{M}_k^\diamond} = [\mathbf{M}_{i:}]_{i \in \mathcal{A}_k^\diamond}$, and $\mathbf{\Tilde{b}}_k^\diamond = [\mathbf{b}_{k,i}]_{i \in \mathcal{A}^\diamond_k}$.
\end{proposition}

\inConf{
\begin{proof}
    See \cite{Pfefferkorn2023}.
\end{proof}
}

\inArxiv{
\begin{proof}
    The result directly follows from applying the method of Lagrange multipliers \cite{Ghojogh2021, Boyd2004} to QP \eqref{eqn_qp}, exploiting sufficiency of the active set \cite{Arnstrom2023}.
    We first formulate the Lagrangian as
    \begin{align}
        \mathcal{L}(\mathbf{u}_k, \mathbf{\mu}_k) = \frac{1}{2} \mathbf{u}_k^\top H \mathbf{u}_k \! + \! h_k^\top \mathbf{u}_k \! + \! r_k \! + \! \bm{\mu}_k^\top (\mathbf{M} \mathbf{u}_k \! - \! \mathbf{b}_k),
    \end{align}
    where $\mathbf{\mu}_k$ is the vector of Lagrange multipliers.
    Applying the Karush-Kuhn-Tucker (KKT) conditions yields
    \begin{subequations}\label{eqn_kkt}
    \begin{align} 
        H \mathbf{u}_k^\diamond + h_k + \mathbf{M}^\top \bm{\mu}^\diamond_k &= \mathbf{0} \label{eq:dLagrandian_du} \\
        \mathbf{M} \mathbf{u}_k^\diamond - \mathbf{b}_k &\leq \mathbf{0} \label{eq:dLagrangian_dmu} \\
        \bm{\mu}^\diamond_k &\geq \mathbf{0} \label{eq:dual_feasibility} \\
        \mu_{k,i}^\diamond (\mathbf{M}_{i:} \mathbf{u}_k^\diamond - \mathbf{b}_{k,i}) &= 0, \quad \forall i \in \llbracket n_x + n_u \rrbracket. \label{eq:complementary_slackness}
    \end{align}
    \end{subequations}
    According to the complementary slackness condition \eqref{eq:complementary_slackness}, we have $\mu_{k,i}^\diamond = 0, \forall i \in \llbracket n_x + n_u \rrbracket \setminus \mathcal{A}_k^\diamond$.
    This enables to remove inactive inequality constraints from the problem formulation (c.f., sufficiency of the active set \cite{Arnstrom2023}).
    Hence, the KKT system \eqref{eqn_kkt} is reduced to a linear system of equations given by
    \begin{equation} \label{eq:reduced_kkt}
        \begin{bmatrix}
            H & \mathbf{\Tilde{M}}_k^{\diamond^\top} \\
            \mathbf{\Tilde{M}}_k^\diamond & \mathbf{0}
        \end{bmatrix}
        \begin{bmatrix}
            \mathbf{u}_k^\diamond \\ \bm{\Tilde{\mu}}_k^\diamond
        \end{bmatrix}
        =
        \begin{bmatrix}
            -h_k \\ \mathbf{\Tilde{b}}_k^\diamond
        \end{bmatrix},
    \end{equation}
    where $\bm{\Tilde{\mu}}_k^\diamond = [\mathbf{\mu}_{k,i}^\diamond]_{i \in \mathcal{A}_k^\diamond}, \mathbf{\Tilde{M}_k^\diamond} = [\mathbf{M}_{i:}]_{i \in \mathcal{A}_k^\diamond}$, $\mathbf{\Tilde{b}}_k^\diamond = [\mathbf{b}_{k,i}]_{i \in \mathcal{A}^\diamond_k}$ and the dual feasibility condition \eqref{eq:dual_feasibility} is trivially fulfilled if $\mathcal{A}_k^\diamond$ is the correct active set.
    Exploiting invertibility of the coefficient matrix of \eqref{eq:reduced_kkt} yields \eqref{eq:qp_solution}.
    Expression \eqref{eq:qp_solution} is guaranteed to be the optimal solution to QP \eqref{eqn_qp} as the KKT conditions are necessary and sufficient under Assumption \ref{as:licq}.
    The solution is unique and global as $H \succ 0$.
\end{proof}
}

The optimal inputs $\mathbf{u}^{\star}_k$ and $\mathbf{u}^\dagger_k$ are obtained using \eqref{eq:qp_solution} when exact and DR constraint tightening are used, respectively.

\inConf{
\textbf{Remark 2:} The active set $\mathcal{A}_k^\diamond$ can be constructed iteratively through additions and removals of constraints to a \textit{working set} $\mathcal{A}_k$ until solution \eqref{eq:qp_solution} is both primal and dual feasible.
This idea is exploited in so-called \textit{active-set methods} for solving QPs \cite{Arnstrom2023}.
Hence, Assumption \ref{as:index_set} is not restrictive and only used to avoid computations that are out of the scope of this work.
As for Assumption \ref{as:licq}, LICQ can be established during control design and is naturally given in many cases.
However, other constraint qualifications can be used, see e.g. \cite{Bergmann2019}.
}

\inArxiv{
\textbf{Remark 2:} The active set $\mathcal{A}_k^\diamond$ can be constructed in an iterative manner through additions and removals of constraints to a \textit{working set} $\mathcal{A}_k$ until primal feasibility \eqref{eq:dLagrangian_dmu} and dual feasibility \eqref{eq:dual_feasibility} are satisfied by solution \eqref{eq:qp_solution}.
This idea is exploited in so-called \textit{active-set methods} for solving QPs \cite{Arnstrom2023}.
Hence, Assumption \ref{as:index_set} is not restrictive and only used to avoid computations that are out of the scope of this work.
As for Assumption \ref{as:licq}, the LICQ criterion is commonly employed in practice as it is a rather weak condition. 
In the considered controller set-up, LICQ can be established during control design and is naturally given in many cases.
However, other constraint qualifications might be used, see e.g. \cite{Bergmann2019}. 
}


\textbf{Remark 3:} 
Given the optimal input sequences $\mathbf{u}_k^{\star}$ and $\mathbf{u}_k^\dagger$ from Proposition \ref{prop:qp_solution}, it can be seen that $\mathfrak{R}_k$ and $\mathfrak{G}_k$ are induced through three main effects:
\begin{enumerate}
    \item[(i)] The active set $\mathcal{A}_k$ differ between the DR and the fully informed controller. The dependencies of $\mathfrak{R}_k$ and $\mathfrak{G}_k$ on the active sets are highly nonlinear.
    \item[(ii)] The initial states $x_{0 \mid k}^{\star}$ and $x_{0 \mid k}^\dagger$ differ as the feasible set of states of the DR controller is smaller than that of the fully informed controller.
    \item[(iii)] The tightening factors $\psi_i, \forall i \in \mathcal{A}_k$ differ between the DR and the fully informed controller.
\end{enumerate}

While $\mathfrak{R}_k$ and $\mathfrak{G}_k$ can be easily computed using the previously stated results, we analyze their behaviors.
First, we start by analyzing $\mathfrak{G}_k$, restricting ourselves to the special time instances defined below.

\begin{definition}
The set $\mathcal{I}$ contains the time steps when the active sets $\mathcal{A}_k^{\star}$ and $\mathcal{A}_k^\dagger$ were same. That is,
\begin{align} \label{eqn_bold_tau}
\mathcal{I} := \left\{ k \in \llbracket N \rrbracket \, \bigg| \, \mathcal{A}_k^{\star} = \mathcal{A}_k^\dagger \right\}. 
\end{align}
\end{definition}
Then, $\forall k \in \mathcal{I},~ \mathbf{\Tilde{M}}_k^{\star} = \mathbf{\Tilde{M}}_k^\dagger = \mathbf{\Tilde{M}}_k$, $\mathbf{V}_k^{\star} = \mathbf{V}_k^\dagger = \mathbf{V}_k$ and the difference between $\mathbf{\Tilde{b}_k^{\star}}$ and $\mathbf{\Tilde{b}_k^\dagger}$ is only due to the tightening factors $\bm{\Tilde{\psi}}^{\star}$ and $\bm{\Tilde{\psi}}^\dagger$ corresponding to $\mathcal{A}_k^{\star}$ and $\mathcal{A}_k^\dagger$.   
\begin{theorem}
    Given $\mathcal{I}$, let $\mathbf{V}_k = \begin{bmatrix} \mathbf{V}_{1,k} & \mathbf{V}_{2,k} \end{bmatrix}, \forall k \in \mathcal{I}$, and let Assumptions \ref{as:index_set} and \ref{as:licq} be satisfied. Then, $\mathfrak{G}_k$ at time steps $k \in \mathcal{I}$ is given by
    \begin{align}
        \begin{split}\label{eq:regret_reform2}
            \mathfrak{G}k \!=\! & -(x_{0 \mid k}^{\star} \! - \! x_{0 \mid k}^\dagger)^\top \Lambda_{1,k} (x_{0 \mid k}^{\star} \! + \! x_{0 \mid k}^\dagger) \\ 
            & - \mathbf{\Tilde{v}}_k^\top \mathrm{diag}(\bm{\Tilde{\psi}}_k^{\star} \! - \! \bm{\Tilde{\psi}}_k^\dagger) \Lambda_{2,k} \mathrm{diag}(\bm{\Tilde{\psi}}_k^{\star} \! + \! \bm{\Tilde{\psi}}_k^\dagger) \mathbf{\Tilde{v}}_k \\
            & + (x_{0 \mid k}^{\star} \! - \! x_{0 \mid k}^\dagger)^\top \Lambda_{3,k} \mathrm{diag}(\bm{\Tilde{\psi}}^{\star}_k \! + \! \bm{\Tilde{\psi}}^\dagger_k) \mathbf{\Tilde{v}}_k \\
            & + (x_{0 \mid k}^{\star} \! + \! x_{0 \mid k}^\dagger)^\top \Lambda_{3,k} \mathrm{diag}(\bm{\Tilde{\psi}}^{\star}_k \! - \! \bm{\Tilde{\psi}}^\dagger_k) \mathbf{\Tilde{v}}_k \\
            & - (x_{0 \mid k}^{\star} \! - \! x_{0 \mid k}^\dagger)^\top \Lambda_{4,k} + \mathbf{\Tilde{v}}^\top_k \mathrm{diag}(\bm{\Tilde{\psi}}^{\star}_k \! - \! \bm{\Tilde{\psi}}^\dagger_k) \Lambda_{5,k},
        \end{split}
    \end{align}
    where
    \begin{align*}
        & \Lambda_{1,k} = \frac{1}{2} \alpha_k^\top H \alpha_k + \frac{1}{2}( \Tilde{h}^\top \alpha_k + \alpha_k^\top \Tilde{h}) + \mathbf{A}^\top \mathbf{QA}, \\
        & \Lambda_{2,k} = \frac{1}{2} \mathbf{V}_{2,k}^\top H \mathbf{V}_{2,k} , \quad \Lambda_{3,k} = \frac{1}{2}(\Tilde{h}^\top \mathbf{V}_{2,k} + \alpha_k^\top H \mathbf{V}_{2,k}) \\
        & \Lambda_{4,k} = \tilde{h}^\top \gamma_k + \alpha_k^\top H \gamma_k, \quad \Lambda_{5,k} = \mathbf{V}_{2,k}^\top H \gamma_k
    \end{align*}
    with $\alpha_k = (V_k \mathbf{\tilde{M}}_k H^{-1} - H^{-1})  \tilde{h} - \mathbf{V}_{2,k} \mathbf{\Tilde{F}}_k \mathbf{A}$, $\Tilde{h} = 2 \mathbf{B}^\top \mathbf{Q} \mathbf{A}$ and $\gamma_k = \mathbf{V}_{1,k} \mathbf{\tilde{d}}_k + \mathbf{V}_{2,k} \mathbf{\Tilde{g}}_k$.
\end{theorem}

\begin{proof}
    Expression \eqref{eq:regret_reform2} is obtained from \eqref{eqn_regretDef_ol} using the cost function representation \eqref{eqn_cost_reform2}, substituting the optimal input sequences $\mathbf{u}^{\star}_k$ and $\mathbf{u}^\dagger_k$ from \eqref{eq:qp_solution} and the definitions of $\mathbf{\tilde{M}}_k$ and $\mathbf{\Tilde{b}}_k$ and performing a series of algebraic operations (not shown here for the brevity of presentation).
\end{proof}

Note that \eqref{eq:regret_reform2} is quadratic in the initial conditions $x_{0 \mid k}^{\star}$, $x_{0 \mid k}^\dagger$ and the tightening factors $\bm{\Tilde{\psi}}_k^{\star}$, $\bm{\Tilde{\psi}}_k^\dagger$ but not directly in their respective differences.

\begin{corollary} \label{corollary_1}
    If no (state and input) constraints are active, then \eqref{eq:regret_reform2} simplifies to
    \begin{align}
    \begin{split}\label{eq:regret_unconstr}
        \mathfrak{G}_k & = -(x_{0 \mid k}^{\star} \! - \! x_{0 \mid k}^\dagger)^\top \Lambda_{1} (x_{0 \mid k}^{\star} \! + \! x_{0 \mid k}^\dagger), \\
        \Lambda_1 & = \mathbf{A}^\top \mathbf{QA} - \frac{1}{2} \Tilde{h}^\top H^{-1} \Tilde{h}
    \end{split}
    \end{align}
    and the optimal input sequences are simply given by the linear feedback laws $\mathbf{u}_k^{\star} = -H^{-1} \tilde{h} x_{0 \mid k}^{\star}$ and $\mathbf{u}_k^\dagger = -H^{-1} \tilde{h} x_{0 \mid k}^\dagger$ from the linear quadratic regulator optimization problem.
\end{corollary}

We continue analyzing the general behavior of $\mathfrak{G}_k$ when no constraints are active, given that the fully informed and the DR SMPC are input-to-state stabilizing in probability according to Assumption \ref{as:ISSp}. 
\begin{assumption}\label{as:phi}
    There exist nonempty sets $\Phi, \Omega^\star, \Omega^\dagger$ such that with $\rho, \varrho$ from Assumption \ref{as:ISSp} it holds that
    \begin{enumerate}
        \item[(i)] $\Phi := \{ x \in \mathcal{X} \mid \| x \| \leq r_\Phi, r_\Phi \in \bbr_{>0} \}$ is such that $\forall x \in \Phi$, no (exactly or distributionally robustly) tightened state constraints are active and no input constraints become active when evaluating the fully informed and the DR SMPC,
        \item[(ii)] $\Omega^\star \! := \! \{ x \in \Phi \mid \| x \| \leq r_{\Omega^\star}, r_{\Omega^\star} = \frac{r_\Phi - \varrho(\| w \|_{L^p})}{M^\star} \} \! \subseteq \! \Phi$,
        \item [(iii)] $\Omega^\dagger \! := \! \{ x \in \Phi \mid \| x \| \leq r_{\Omega^\dagger}, r_{\Omega^\dagger} = \frac{r_\Phi - \varrho(\| w \|_{L^p})}{M^\dagger} \} \! \subseteq \! \Phi$,
    \end{enumerate}
    and $r_{\Omega^\star} \! > \! \frac{\rho(\| w \|_{L^p})}{\omega^\star} \! + \! \epsilon$, $r_{\Omega^\dagger} \! > \! \frac{\rho(\| w \|_{L^p})}{\omega^\dagger} \! + \! \epsilon$ for some $\epsilon \! > \! 0$.
    
\end{assumption}

\inConf{The next lemma proves the recurrent nature (see appendix of \cite{Pfefferkorn2023} for the definition) of the sets $\Phi, \Omega^\star$ and $\Omega^\dagger$ as well as boundedness in probability of system \eqref{eqn_system_dynamics} to $\Phi$ under both the fully informed and the DR SMPC.}
\inArxiv{The next lemma proves the recurrent nature (see appendix for the definition) of the sets $\Phi, \Omega^\star$ and $\Omega^\dagger$ as well as boundedness in probability of system \eqref{eqn_system_dynamics} to $\Phi$ under both the fully informed and the DR SMPC.}

\begin{lemma}\label{lemma:phi_rpi}
    Under Assumptions \ref{as:ISSp}-\ref{as:phi}, the sets $\Phi$ and  $\Omega^\diamond$ are recurrent for system \eqref{eqn_system_dynamics} under the fully informed ($\diamond \equiv \star$) and the DR ($\diamond \equiv \dagger$) SMPC.
    Furthermore, if $x^\diamond_k \in \Omega^\diamond$, then system \eqref{eqn_system_dynamics} satisfies    
    \begin{align}\label{eq:phi_bounded_in_p}
        \mathbb{P} [x^\diamond_{k+i} \in \Phi, \forall i \leq \bar{K}] \geq 1 - \varepsilon.
    \end{align}
\end{lemma}


\begin{proof}
    As the optimal value functions \eqref{eq:optimal_value_functions} are eISSp Lyapunov functions for the respective closed-loop systems by Assumption \ref{as:ISSp}, any sublevel sets $\mathcal{V}_{\gamma^\star} \! := \! \{ x \! \in \! \mathcal{X} \! \mid \! J^\star_\mathrm{SMPC}(x) \! \leq \! \gamma^\star \}, \gamma^\star > \frac{\rho(\| w \|_{L^p})}{w^\star}$ and $\mathcal{V}_{\gamma^\dagger} \! := \! \{ x \! \in \! \mathcal{X} \! \mid \! J^\dagger_\mathrm{SMPC}(x) \! \leq \! \gamma^\dagger \}, \gamma^\dagger > \frac{\rho(\| w \|_{L^p})}{w^\dagger}$ are recurrent, see \cite{Culbertson2023}.
    Since $r_{\Omega^\star} > \frac{\rho(\| w \|_{L^p})}{w^\star} + \epsilon$ and $r_{\Omega^\dagger} > \frac{\rho(\| w \|_{L^p})}{w^\dagger} + \epsilon$, there exist $\bar{\gamma}^\star > \frac{\rho(\| w \|_{L^p})}{w^\star}$ and $\bar{\gamma}^\dagger > \frac{\rho(\| w \|_{L^p})}{w^\dagger}$ such that $\mathcal{V}_{\bar{\gamma}^{\star}} \subset \Omega^\star, \mathcal{V}_{\bar{\gamma}^\dagger} \subset \Omega^\dagger$. Hence, $\Omega^\star, \Omega^\dagger$ are recurrent and since $\Omega^\star \subseteq \Phi, \Omega^\dagger \subseteq \Phi$, $\Phi$ is recurrent under both the fully informed and the DR controller. Expression \eqref{eq:phi_bounded_in_p} follows directly from \eqref{eq:eISSp_by_assumption} as $ \|x^\diamond_{k+i} \| \leq \beta^\diamond(\|x_k^\diamond\|, i) + \varrho(\|w_{k+i}\|_{L^p}) \leq \beta^\diamond(\|x_k^\diamond\|, 0) + \varrho(\|w_{k+i}\|_{L^p})$ with probability at least $1 - \varepsilon$ for a finite horizon $i \leq \bar{K}$ and $\beta^\diamond(\|x_k^\diamond\|, 0) + \varrho(\|w_{k+i}\|_{L^p}) \leq r_\Phi$ for $x_k^\diamond \in \Omega^\diamond$.
    %
\end{proof}
Lemma \ref{lemma:phi_rpi} states that the system visits $\Phi$ in finite time under both controllers.
Furthermore, once the system has entered $\Phi$, it will remain in $\Phi$ for a finite horizon with high probability.
\inArxiv{That is, $\Phi$ takes the role of a probabilistic invariant set.}
We now study the system's behavior while in $\Phi$.

\begin{lemma}\label{lemma:convergence_of_error}
    Under Assumptions \ref{as:ISSp}-\ref{as:phi}, there exists $\tilde{\beta} \in \mathcal{KL}$ such that if $x_k^{\star} \in \Omega^\star$ and $x_k^\dagger \in \Omega^\dagger$, it holds $\forall i \leq \bar{K}$ that
    \begin{equation}\label{eq:error_system_convergence}
        \mathbb{P} [ \lVert x_{k+i}^{\star} - x_{k+i}^\dagger \rVert \leq \tilde{\beta}( \lVert x_k^{\star} \rVert + \lVert x_k^\dagger \rVert, i ) ] \geq \max \{ 0, 1 - 2 \varepsilon \}.
    \end{equation}
\end{lemma}


\inConf{
\begin{proof}
    Suppose that $x_k^{\star} \in \Omega^\star$ and $x_k^\dagger \in \Omega^\dagger$ and set $\bar{x}_k^\star = x_k^\star$ and $\bar{x}_k^\dagger = x_k^\dagger$. Then, $\mathbb{P}[x_{k+1}^\star \in \Phi, \forall i \leq \bar{K}] \geq 1 - \varepsilon$ and $\mathbb{P}[x_{k+i}^\dagger \in \Phi, \forall i \leq \bar{K}] \geq 1 - \varepsilon$ by Lemma \ref{lemma:phi_rpi}.
    If $x_{k+i}^{\star}, x_{k+i}^\dagger \in \Phi$ for all $i \leq \bar{K}$, no constraints are active by Assumption \ref{as:phi} and the MPC feedback law reduces the linear quadratic regulator feedback law for both the fully informed and the DR case by Corollary \ref{corollary_1}.
    Hence, the error dynamics between the fully informed and the DR case is given by $e_{k+i} := x^\star_{k+i} - x^\dagger_{k+i} = (A - BH^{-1} \tilde{h}) e_{k+i-1}$.
    As the same disturbances will act in both the fully informed and the DR case by Assumption \ref{as:same_disturbance}, we can equivalently write the error in terms of the nominal systems as $e_{k+i} = (A - BH^{-1}\tilde{h}) (\bar{x}_{k+i-1}^\star - \bar{x}_{k+i-1}^\dagger)$.
    Exploiting stability according to Assumption \ref{as:ISSp} and Corollary \ref{corollary:nominal_stability}, the state trajectories of the fully informed and the DR converge to each other and we can bound the error from above by
    \begin{align}
        \| e_{k+i} \| & = \| \bar{x}_{k+i}^\star - \bar{x}_{k+i}^\dagger \| \leq \| \bar{x}_{k+i}^\star \| + \| \bar{x}_{k+i}^\dagger \| \nonumber \\ 
        & \leq \beta^\star( \| x_{k}^\star \|, i ) + \beta^\dagger( \| x_{k}^\dagger \|, i) \nonumber \\
        & = M^\star {(\nu^\star)}^i \| \bar{x}_k^\star \| + M^\dagger {(\nu^\dagger)}^i \| \bar{x}_k^\dagger \| \nonumber \\
        & \leq \max \{ M^\star, M^\dagger \} (\max \{ \nu^\star, \nu^\dagger \})^i (\| \bar{x}_k^\star \| + \| \bar{x}_k^\dagger \|) \nonumber \\
        & =: \tilde{\beta}(\| \bar{x}_k^\star \| + \| \bar{x}_k^\dagger \|, i) \label{eq:auxiliary_bound}
    \end{align}
    with $\beta^{\star}, \beta^\dagger$ from Assumption \ref{as:ISSp}.
    Clearly, $\tilde{\beta} \in \mathcal{KL}$.
    As the above considerations require both systems to be in $\Phi$ for all $i \leq \bar{K}$, \eqref{eq:auxiliary_bound} holds only true with probability $\mathbb{P}[(x^\star_{k+i} \in \Phi, \forall i \leq \bar{K}) \wedge (x^\dagger_{k+i} \in \Phi, \forall i \leq \bar{K})]$. Although $\mathbb{P}[x_{k+1}^\star \in \Phi, \forall i \leq \bar{K}] \geq 1 - \varepsilon$ and $\mathbb{P}[x_{k+i}^\dagger \in \Phi, \forall i \leq \bar{K}] \geq 1 - \varepsilon$ are known, both events need to be considered dependent as the fully informed and the DR case rely on the same disturbances. To lower bound their joint probability, we apply the Fréchet inequality, yielding
    \begin{align*}
        & \mathbb{P}[(x^\star_{k+i} \in \Phi, \forall i \leq \bar{K}) \wedge (x^\dagger_{k+i} \in \Phi, \forall i \leq \bar{K})] \\
        & ~~~~~~~~~~~~~~~~ \geq \max \{ 0, 2(1 - \varepsilon) - 1 \} = \max \{0, 1 - 2 \varepsilon \}.
    \end{align*}
    This concludes the proof.
\end{proof}
}

\inArxiv{
\begin{proof}
    Suppose that $x_k^{\star} \in \Omega^\star$ and $x_k^\dagger \in \Omega^\dagger$ and set $\bar{x}_k^\star = x_k^\star$ and $\bar{x}_k^\dagger = x_k^\dagger$. Then, $\mathbb{P}[x_{k+1}^\star \in \Phi, \forall i \leq \bar{K}] \geq 1 - \varepsilon$ and $\mathbb{P}[x_{k+i}^\dagger \in \Phi, \forall i \leq \bar{K}] \geq 1 - \varepsilon$ by Lemma \ref{lemma:phi_rpi}.
    If $x_{k+i}^{\star}, x_{k+i}^\dagger \in \Phi$ for all $i \leq \bar{K}$, no constraints are active by Assumption \ref{as:phi} and the MPC feedback law reduces the linear quadratic regulator feedback law for both the fully informed and the DR case by Corollary \ref{corollary_1}.
    Hence, the error dynamics between the fully informed and the DR case is given by $e_{k+i} := x^\star_{k+i} - x^\dagger_{k+i} = (A - BH^{-1} \tilde{h}) e_{k+i-1}$.
    As the same disturbances will act in both the fully informed and the DR case by Assumption \ref{as:same_disturbance}, we can equivalently write the error in terms of the nominal systems as $e_{k+i} = (A - BH^{-1}\tilde{h}) (\bar{x}_{k+i-1}^\star - \bar{x}_{k+i-1}^\dagger)$.
    Exploiting stability according to Assumption \ref{as:ISSp} and Corollary \ref{corollary:nominal_stability}, the state trajectories of the fully informed and the DR converge to each other and we can bound the error from above by
    \begin{align}
        \| e_{k+i} \| = \| \bar{x}_{k+i}^\star - \bar{x}_{k+i}^\dagger \| & \leq \| \bar{x}_{k+i}^\star \| + \| \bar{x}_{k+i}^\dagger \| \nonumber \\ 
        & \leq \beta^\star( \| x_{k}^\star \|, i ) + \beta^\dagger( \| x_{k}^\dagger \|, i) \nonumber \\
        & = M^\star {(\nu^\star)}^i \| \bar{x}_k^\star \| + M^\dagger {(\nu^\dagger)}^i \| \bar{x}_k^\dagger \| \nonumber \\
        & \leq \max \{ M^\star, M^\dagger \} (\max \{ \nu^\star, \nu^\dagger \})^i (\| \bar{x}_k^\star \| + \| \bar{x}_k^\dagger \|) \nonumber \\
        & =: \tilde{\beta}(\| \bar{x}_k^\star \| + \| \bar{x}_k^\dagger \|, i) \label{eq:auxiliary_bound}
    \end{align}
    with $\beta^{\star}, \beta^\dagger$ from Assumption \ref{as:ISSp}.
    Clearly, $\tilde{\beta} \in \mathcal{KL}$.
    As the above considerations require both systems to be in $\Phi$ for all $i \leq \bar{K}$, \eqref{eq:auxiliary_bound} holds only true with probability $\mathbb{P}[(x^\star_{k+i} \in \Phi, \forall i \leq \bar{K}) \wedge (x^\dagger_{k+i} \in \Phi, \forall i \leq \bar{K})]$. Although $\mathbb{P}[x_{k+1}^\star \in \Phi, \forall i \leq \bar{K}] \geq 1 - \varepsilon$ and $\mathbb{P}[x_{k+i}^\dagger \in \Phi, \forall i \leq \bar{K}] \geq 1 - \varepsilon$ are known, both events need to be considered dependent as the fully informed and the DR case rely on the same disturbances. To lower bound their joint probability, we apply the Fréchet inequality, yielding
    \begin{align*}
        \mathbb{P}[(x^\star_{k+i} \in \Phi, \forall i \leq \bar{K}) \wedge (x^\dagger_{k+i} \in \Phi, \forall i \leq \bar{K})] & \geq \max \{ 0, \mathbb{P}[x_{k+1}^\star \in \Phi, \forall i \leq \bar{K}] + \mathbb{P}[x_{k+1}^\dagger \in \Phi, \forall i \leq \bar{K}] - 1 \} \\
        & \geq \max \{ 0, 2(1 - \varepsilon) - 1 \} = \max \{0, 1 - 2 \varepsilon \}.
    \end{align*}
    This concludes the proof.
\end{proof}
}

By Lemma \ref{lemma:convergence_of_error}, the states of the system under the DR and the fully informed controller, respectively, converge to each other while in $\Phi$. 
We state the following theorem about $\mathfrak{R}_k^\mathrm{ol}$.

\begin{theorem}\label{theorem_2}
    Let Assumptions \ref{as:ISSp}-\ref{as:phi} hold.
    If $x_k^{\star} \in \Omega^\star$ and $x_k^\dagger \in \Omega^\dagger$, then $\exists \sigma \in \mathcal{KL}$ such that
    \begin{equation}
        \mathbb{P}[ \lVert \mathfrak{G}_{k+i} \rVert \! \leq \! \sigma(\lVert x_k^{\star} \rVert \! + \! \lVert x_k^\dagger \rVert , i ), \forall i \! \leq \! \bar{K} ] \! \geq \! \max \{ 0, 1 \! - \! 2\varepsilon \}.
    \end{equation}
\end{theorem}

\inConf{
\begin{proof}
    If $x_k^{\star} \in \Omega^\star$ and $x_k^\dagger \in \Omega^\dagger$, both systems remain jointly in $\Phi$ for horizon $\bar{K}$ with probability at least $\max \{ 0,1-2\varepsilon \}$ exploiting Fréchet bounds and the individual probabilities \eqref{eq:phi_bounded_in_p}.
    As no constraints are active for $x \in \Phi$, we find that $\forall i \leq \bar{K}$
    \begin{align*}
        & \mathbb{P} [\mathfrak{G}_{k+i} = \eqref{eq:regret_unconstr}] \! \geq \! \max \{ 0 , 1 \! - \! 2 \varepsilon \} \\
        \Leftrightarrow ~ & \mathbb{P} [\lVert \mathfrak{G}_{k+i} \rVert \! \leq \! \lVert x_{k+i}^{\star} \! - \! x_{k+i}^\dagger \rVert \lVert \Lambda_1 \rVert \lVert x_{k+i}^{\star} \! + \! x_{k+i}^\dagger \rVert] \\
        & \qquad \qquad \qquad \qquad \qquad \qquad \qquad ~~~ \geq \! \max \{ 0 , 1 \! - \! 2 \varepsilon \} \\
        \Leftrightarrow ~ & \mathbb{P}[ \lVert \mathfrak{G}_{k+i} \rVert \! \leq \! \tilde{\beta}(\lVert x_k^{\star} \rVert \! + \! \lVert x_k^\dagger \rVert, i) 2 \lVert \Lambda_1 \rVert r_\Phi ] \\
        & \qquad \qquad \qquad \qquad \qquad \qquad \qquad \quad  \geq \! \max \{ 0 , 1 \! - \! 2 \varepsilon \},
    \end{align*}
    where $\lVert x_{k+i}^{\star} \! + \! x_{k+i}^\dagger \rVert \leq \lVert x_{k+i}^{\star} \rVert \! + \! \lVert x_{k+i}^\dagger \rVert \leq 2 r$ as $x_{k+i}^{\star}, x_{k+i}^\dagger \in \Phi$.
    Clearly, $2 \lVert \Lambda_1 \rVert r > 0$ and hence, $\sigma( \lVert x_k^{\star} \rVert \! + \! \lVert x_k^\dagger \rVert, i) = 2 \lVert \Lambda_1 \rVert r \tilde{\beta}(\lVert x_k^{\star} \rVert \! + \! \lVert x_k^\dagger \rVert, i)$ is a class $\mathcal{KL}$ function.
\end{proof}}

\inArxiv{
\begin{proof}
    If $x_k^{\star} \in \Omega^\star$ and $x_k^\dagger \in \Omega^\dagger$, both systems remain jointly in $\Phi$ for horizon $\bar{K}$ with probability at least $\max \{ 0,1-2\varepsilon \}$ exploiting Fréchet bounds and the individual probabilities \eqref{eq:phi_bounded_in_p}.
    As no constraints are active for $x \in \Phi$, we find that $\forall i \leq \bar{K}$
    \begin{align*}
        & \mathbb{P} [\mathfrak{G}_{k+i} = \eqref{eq:regret_unconstr}] \! \geq \! \max \{ 0 , 1 \! - \! 2 \varepsilon \} \\
        \Leftrightarrow ~ & \mathbb{P} [\lVert \mathfrak{G}_{k+i} \rVert \! \leq \! \lVert x_{k+i}^{\star} \! - \! x_{k+i}^\dagger \rVert \lVert \Lambda_1 \rVert \lVert x_{k+i}^{\star} \! + \! x_{k+i}^\dagger \rVert] \! \geq \! \max \{ 0 , 1 \! - \! 2 \varepsilon \} \\
        \Leftrightarrow ~ & \mathbb{P}[ \lVert \mathfrak{G}_{k+i} \rVert \! \leq \! \tilde{\beta}(\lVert x_k^{\star} \rVert \! + \! \lVert x_k^\dagger \rVert, i) 2 \lVert \Lambda_1 \rVert r_\Phi ] \! \geq \! \max \{ 0 , 1 \! - \! 2 \varepsilon \},
    \end{align*}
    where $\lVert x_{k+i}^{\star} \! + \! x_{k+i}^\dagger \rVert \leq \lVert x_{k+i}^{\star} \rVert \! + \! \lVert x_{k+i}^\dagger \rVert \leq 2 r$ as $x_{k+i}^{\star}, x_{k+i}^\dagger \in \Phi$.
    Clearly, $2 \lVert \Lambda_1 \rVert r > 0$ and hence, $\sigma( \lVert x_k^{\star} \rVert \! + \! \lVert x_k^\dagger \rVert, i) = 2 \lVert \Lambda_1 \rVert r \tilde{\beta}(\lVert x_k^{\star} \rVert \! + \! \lVert x_k^\dagger \rVert, i)$ is a class $\mathcal{KL}$ function.
\end{proof}}

Theorem \ref{theorem_2} states that $\mathfrak{G}_k$ decreases with high probability over a finite number of subsequent time steps from $k$ on.
That is, $\mathfrak{G}_k$ will converge to zero as long as no disturbance realization causes the system states $x_k^{\star}$ or $x_k^\dagger$ to leave $\Phi$.

\inConf{
We now analyze $\mathfrak{R}_k$.
To this end, we express the closed-loop input $u^\diamond_k = u_{0 \mid k}^\diamond$ (the first element of \eqref{eq:qp_solution}) and state realization $\hat{x}_k^\diamond$ explicitly in terms of the initial state $x_0$ and the disturbance realizations $(\hat{w}_i)_{i \in [0:k-1]}$ as
\begin{multline}\label{eq:closed_loop_x}
    \hat{x}_k^\diamond = \left( A^k + \sum_{i=0}^{k-1} A^i B \Psi_{k-1-i}^\diamond \right) x_0 \\ + \sum_{i=0}^{k-1} A^i B \gamma_{k-1-i}^\diamond + \sum_{i=0}^{k-1} A^i \hat{w}_{k-1-i},
\end{multline}
\vspace{-7pt}
\begin{align}\label{eq:closed_loop_u}
    u_k^\diamond = \underbrace{P^\diamond_k \left( A^k + \sum_{i=0}^{k-1} A^i B \Psi_{k-1-i}^\diamond \right)}_{=:\Psi_k^\diamond} x_0 + \gamma^\diamond_k.
\end{align}
Therein, the closed-loop input is expressed as $u_k^\diamond = P_k^\diamond \hat{x}_k^\diamond + q_k^\diamond$ with $P_k^\diamond$ and $q_k^\diamond$ implicitly defined via \eqref{eq:qp_solution} and
\begin{align}
    \gamma_k^\diamond = P_k^\diamond \left( \sum_{i=0}^{k-1} A^i B \gamma_{k-1-i}^\diamond + \sum_{i=0}^{k-1} A^i \hat{w}_{k-1-i} \right) + q_k^\diamond.
\end{align}
Using \eqref{eq:closed_loop_x} and \eqref{eq:closed_loop_u} in \eqref{eqn_regretDef_cl} enables further analysis.
}
\inArxiv{
We now analyze $\mathfrak{R}_k$.
To this end, we express the closed-loop input $u^\diamond_k = u_{0 \mid k}^\diamond$ (the first element of \eqref{eq:qp_solution}) and state realization $\hat{x}_k^\diamond$ explicitly in terms of the initial state $x_0$ and the disturbance realizations $(\hat{w}_i)_{i \in [0:k-1]}$ as
\begin{align}
    \hat{x}_k^\diamond &= \left( A^k + \sum_{i=0}^{k-1} A^i B \Psi_{k-1-i}^\diamond \right) x_0 + \sum_{i=0}^{k-1} A^i B \gamma_{k-1-i}^\diamond + \sum_{i=0}^{k-1} A^i \hat{w}_{k-1-i}, \label{eq:closed_loop_x} \\
    u_k^\diamond &= \underbrace{P^\diamond_k \left( A^k + \sum_{i=0}^{k-1} A^i B \Psi_{k-1-i}^\diamond \right)}_{=:\Psi_k^\diamond} x_0 + \gamma^\diamond_k. \label{eq:closed_loop_u}
\end{align}
Therein, the closed-loop input is expressed as $u_k^\diamond = P_k^\diamond \hat{x}_k^\diamond + q_k^\diamond$ with $P_k^\diamond$ and $q_k^\diamond$ implicitly defined via \eqref{eq:qp_solution} and
\begin{align}
    \gamma_k^\diamond = P_k^\diamond \left( \sum_{i=0}^{k-1} A^i B \gamma_{k-1-i}^\diamond + \sum_{i=0}^{k-1} A^i \hat{w}_{k-1-i} \right) + q_k^\diamond.
\end{align}
Using \eqref{eq:closed_loop_x} and \eqref{eq:closed_loop_u} in \eqref{eqn_regretDef_cl} enables further analysis.
}


\inConf{
\begin{theorem}\label{theorem:closed_loop_regret}
    Let Assumptions \ref{as:ISSp}-\ref{as:phi} hold and let $\hat{x}^{\star}_{k-\kappa} \in \Omega^\star$ and $\hat{x}^\dagger_{k-\kappa} \in \Omega^\dagger$ for some $\kappa \in \bbn$. 
    Further, let the disturbance realization $(\hat{w}_{k-\kappa+i})_{i \in [0:\kappa-1]}$ be such that $\forall i \in \llbracket \kappa \rrbracket: \hat{x}_{k-\kappa+i}^{\star}, \hat{x}_{k-\kappa+i}^\dagger \in \Phi$.
    Then, $\exists \lambda \in \mathcal{KL}$ such that $\forall i \in \llbracket \kappa \rrbracket$ it holds that
    \begin{multline}
        \lVert \mathfrak{R}_{k-\kappa+i} \rVert \leq \lVert \mathfrak{R}_{k-\kappa+i-1} \rVert \\ + \lambda( \lVert \hat{x}_{k-\kappa}^{\star} \rVert + \lVert \hat{x}_{k-\kappa}^\dagger \rVert, i ).
    \end{multline}
\end{theorem}

\begin{proof}
    The assertion follows from
    \begin{enumerate}
        \item[(i)] applying the triangle inequality to \eqref{eqn_regretDef_cl},
        \item[(ii)] upper bounding $\lVert \mathfrak{R}_{k-\kappa+i} \rVert - \lVert \mathfrak{R}_{k-\kappa+i-1} \rVert$ in terms of $\lVert x_{k-\kappa+i}^{\star} - x_{k-\kappa+i}^\dagger \rVert$,
        \item[(iii)] using Lemma \eqref{lemma:convergence_of_error} to upper bound $\lVert x_{k-\kappa+i}^{\star} - x_{k-\kappa+i}^\dagger \rVert$ by $\tilde{\beta}(\lVert \hat{x}_{k-\kappa}^{\star} \rVert + \lVert \hat{x}_{k-\kappa}^\dagger \rVert, i)$, and
        \item[(iv)] rewriting this upper bound by a function $\lambda \in \mathcal{KL}$.
    \end{enumerate}
    We omit details for brevity and refer the reader to \cite{Pfefferkorn2023}.
\end{proof}
}

\inArxiv{
\begin{theorem}\label{theorem:closed_loop_regret}
    Let Assumptions \ref{as:ISSp}-\ref{as:phi} hold and let $\hat{x}^{\star}_{k-\kappa} \in \Omega^\star$ and $\hat{x}^\dagger_{k-\kappa} \in \Omega^\dagger$ for some $\kappa \in \bbn$. 
    Further, let the disturbance realization $(\hat{w}_{k-\kappa+i})_{i \in [0:\kappa-1]}$ be such that $\forall i \in \llbracket \kappa \rrbracket: \hat{x}_{k-\kappa+i}^{\star}, \hat{x}_{k-\kappa+i}^\dagger \in \Phi$.
    Then, $\exists \lambda \in \mathcal{KL}$ such that $\forall i \in \llbracket \kappa \rrbracket$ it holds that
    \begin{equation}
        \lVert \mathfrak{R}_{k-\kappa+i} \rVert \leq \lVert \mathfrak{R}_{k-\kappa+i-1} \rVert + \lambda( \lVert \hat{x}_{k-\kappa}^{\star} \rVert + \lVert \hat{x}_{k-\kappa}^\dagger \rVert, i ).
    \end{equation}
\end{theorem}

\begin{proof}
    According to \eqref{eqn_regretDef_cl}, the closed-loop regret over the finite horizon $[k-\kappa, k]$ is given by
    \begin{equation}
        \mathfrak{R}_{k-\kappa+i} = \mathfrak{R}_{k-\kappa+i-1} + \lVert \hat{x}^\dagger_{k-\kappa+i} \rVert^2_Q - \lVert \hat{x}^{\star}_{k-\kappa+i} \rVert^2_Q + \lVert u^\dagger_{k-\kappa+i} \rVert^2_R - \lVert u^{\star}_{k-\kappa+i} \rVert^2_R.
    \end{equation}
    Applying the triangle inequality then yields
    \begin{equation}\label{eq:cl_regret_bound_finite_horizon}
        \lVert \mathfrak{R}_{k-\kappa+i} \rVert \leq \lVert \mathfrak{R}_{k-\kappa+i-1} \rVert + \lVert \hat{x}^\dagger_{k-\kappa+i} - \hat{x}^{\star}_{k-\kappa+i} \rVert^2_Q  + \lVert u^\dagger_{k-\kappa+i} - u^{\star}_{k-\kappa+i} \rVert^2_R.
    \end{equation}
    Lemma \ref{lemma:convergence_of_error} implies that $\forall i \in \llbracket \kappa \rrbracket$
    \begin{equation}\label{eq:error_closed_loop}
        \lVert x_{k-\kappa+i}^{\star} - x_{k-\kappa+i}^\dagger \rVert \leq \tilde{\beta}( \lVert \hat{x}_{k-\kappa}^{\star} \rVert + \lVert \hat{x}_{k-\kappa}^\dagger \rVert, i ).
    \end{equation}
    Exploiting \eqref{eq:error_closed_loop} and Corollary \ref{corollary_1}, we find that
    \begin{align}
        \lVert u_{k-\kappa+i}^\dagger - u^{\star}_{k-\kappa+i} \rVert^2_R & \leq \lVert \tilde{h}^\top H^{-1} R H^{-1} \tilde{h} \rVert \cdot \tilde{\beta}^2( \lVert \hat{x}^{\star}_{k-\kappa} \rVert + \lVert \hat{x}_{k-\kappa}^\dagger \rVert, i), \label{eq:diff_u_bound} \\
        \lVert \hat{x}^\dagger_{k-\kappa+i} - \hat{x}^{\star}_{k-\kappa+i} \rVert_Q^2 & \leq \lVert Q \rVert \cdot \tilde{\beta}^2( \lVert \hat{x}^{\star}_{k-\kappa} \rVert + \lVert \hat{x}_{k-\kappa}^\dagger \rVert, i). \label{eq:diff_x_bound}
    \end{align}
    Substituting \eqref{eq:diff_u_bound} and \eqref{eq:diff_x_bound} in \eqref{eq:cl_regret_bound_finite_horizon}, we obtain
    \begin{equation}
        \lVert \mathfrak{R}_{k-\kappa+i} \rVert \leq \lVert \mathfrak{R}_{k-\kappa+i-1} \rVert + \lambda (\lVert \hat{x}^{\star}_{k-\kappa} \rVert + \lVert \hat{x}_{k-\kappa}^\dagger \rVert, i)
    \end{equation}
    with $\lambda (\lVert \hat{x}^{\star}_{k-\kappa} \rVert + \lVert \hat{x}_{k-\kappa}^\dagger \rVert, i) = (\lVert \tilde{h}^\top H^{-1} R H^{-1} \tilde{h} \rVert + \lVert Q \rVert) \tilde{\beta}^2( \lVert \hat{x}^{\star}_{k-\kappa} \rVert + \lVert \hat{x}_{k-\kappa}^\dagger \rVert, i)$. We see, by the construction of $\tilde{\beta}$ in Lemma \ref{lemma:convergence_of_error}, that $\lambda \in \mathcal{KL}$.
\end{proof}
}

By Theorem \ref{theorem:closed_loop_regret}, loosely speaking, $\mathfrak{R}_k$ converges to a constant value while the system states are in $\Phi$ for both the DR and the fully informed case.
Since $\mathfrak{R}_k$ and $\mathfrak{G}_k$ are mainly induced by active, excessively tightened state constraints, those results indicate the potential of dynamically allocating the individual risk budgets $\delta_i$ to the constraints by leveraging the insights obtained from regret analysis.


\section{Numerical Simulation} \label{sec_num_sim}


\begin{figure}
    \centering
    \input{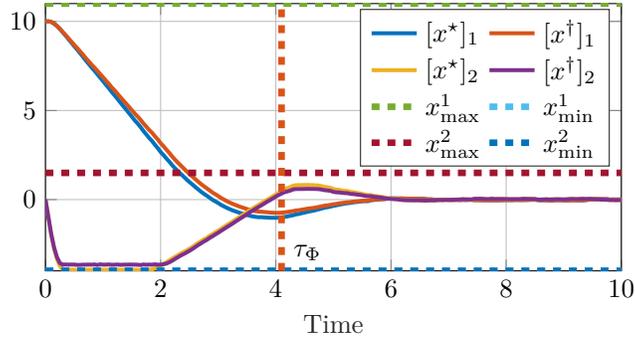}
    \caption{Simulation results for the DR and the fully informed case.}
    \label{fig_1}
\end{figure}

\begin{figure}
    \centering
    \input{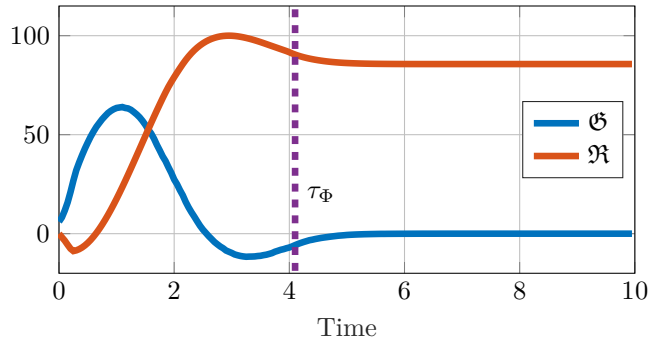}
    \caption{Evolution of closed-loop regret $\mathfrak{R}_k$ and suboptimality gap $\mathfrak{G}_k$ with respect to time.}
    \label{fig_2}
    \vspace{-10pt}
\end{figure}

We consider a discretized double integrator with discretization time step $dt = 0.05$ given by
\begin{align} \label{eqn_simulation_dynamics}
    A = \begin{bmatrix} 1 & dt \\ 0 & 1
    \end{bmatrix}, 
    B = \begin{bmatrix} 0 \\ dt
    \end{bmatrix}, 
    C = I_{2},  
    D = 0_{2 \times 1}. 
\end{align}
The model was simulated for a total of $T = 200$ time steps with a total risk budget of $\Delta = 0.1$. The prediction horizon was set to be $N = 5$ time steps. The process noise was sampled from a multivariate Laplacian distribution with zero mean and covariance of $\mathbf{\Sigma}^{\mathbf{w}} = {0.01}^{2} I_{2}$. The state and control penalty matrices are $Q = \mathrm{diag}(1, 0.1)$ and $R = 0.1$ respectively. The control and the state constraints are given by $u \in [-20, 2]$ and $x \in [-4, 11] \times [-4, 1.5]$.\footnote{The code is available under: \url{https://github.com/mpfefferk/Regret-and-Conservatism-of-SMPC}.}

The results of simulating the system given by \eqref{eqn_simulation_dynamics} using both the fully informed and the DR SMPC controller are shown in Figures \ref{fig_1} and \ref{fig_2}. 
Specifically, the states are plotted in Figure \ref{fig_1} and $\mathfrak{R}_k$ and $\mathfrak{G}_k$ associated with the DR SMPC controller are plotted in Figure \ref{fig_2}.
In the beginning, $\mathfrak{R}_k$ is negative as it is a short-sighted measure: initially, the increased cautiousness of the DR SMPC leads to lower closed-loop costs when compared to the fully informed SMPC.
At the same time, $\mathcal{G}_k$ increases rapidly: the DR SMPC is anticipated to perform worse in the future.
While at later times $\mathfrak{R}_k$ is positive, $\mathfrak{G}_k$ becomes negative.
The latter is because the system under the cautious DR controller is at that time closer to the reference than its fully informed counterpart and hence, its remaining cost-to-go is lower.
At $t = 4.1 s$, all constraints become inactive for both the DR and the fully informed controller, indicating that both systems have entered the (here implicitly defined) set $\Phi$.
From that time point on until the end of the simulation, we observe the convergence behavior of $\mathfrak{R}_k$ and $\mathfrak{G}_k$ proven in Theorems \ref{theorem_2} and \ref{theorem:closed_loop_regret}.


\section{Conclusion} \label{sec_conclusion}

A linear DR SMPC formulation exploiting moment-based ambiguity sets has been analyzed in regard of conservatism and regret.
We have analyzed the different aspects of regret, which together create the overall picture of its behavior.
Especially, regret was shown to converge to constant values over finite time periods and to increase only in between two such periods.
These findings were underlined by simulations.
Future research will be dedicated towards applying the established framework to extended SMPC formulations that may include online estimation or learning of unknown quantities.
Furthermore, the effect of non-uniform risk allocation on conservatism and regret is to be further investigated.

\inArxiv{
\section*{Acknowledgment}
This project has received funding from the European Research Council (ERC) under the European Union’s Horizon 2020 research and innovation program under grant agreement No 834142 (Scalable Control) and from the German Research Foundation (Research Training Group 2297). V. Renganathan is a member of the ELLIIT Strategic Research Area in Lund University.
}

\inConf{
\bibliographystyle{IEEEtran}
\bibliography{references}
}
\inArxiv{
\bibliographystyle{tmlr}
\bibliography{references}
}

\inArxiv{
\section*{Appendix}

We begin by stating the definitions of class $\mathcal{K}$ and class $\mathcal{KL}$ functions which are being used to prove the stability of the system that we consider. 

\begin{definition}
A scalar continuous function $\alpha(r)$, defined for $r \in [0, a), a > 0$ is said to belong to class $\mathcal{K}$ if it is strictly increasing and $\alpha(0) = 0$. It is said to belong to class $\mathcal{K}_{\infty}$ if it defined $\forall r \geq 0$ and $\alpha(r) \rightarrow \infty$ as $r \rightarrow \infty$.
\end{definition}
\begin{definition}
A scalar continuous function $\beta(r, s)$, defined for $r \in [0,a), a > 0$ and $s \in [0,\infty)$ is said to belong to class $\mathcal{KL}$ if, for each fixed s, the mapping $\beta(r, s)$ belongs to class $\mathcal{K}$ with respect to $r$ and, for each fixed $r$, the mapping $\beta(r, s)$ is decreasing with respect to $s$ and $\beta(r, s) \rightarrow 0$ as $s \rightarrow \infty$.
\end{definition}

Next, we define the $L^p$ norm of a random variable.

\begin{definition}(From \cite{Culbertson2023})
    A random variable $w \sim \mathbb{P}_w$ belongs to $L^p$, denoted by $w \! \in \! L^p$, for some $p \! > \! 0$ if $\lVert w \rVert_{L^p} := \mathbb{E} \left [ \lVert w \rVert^p \right ]^{\frac{1}{p}} < \infty$, where $\lVert \cdot \rVert$ defines a typical norm on $\bbr^{n}$.
\end{definition}

\textcolor{black}{Note that in the considered set-up, we assumed that $w_k$ has finite mean and (co-)variance, implying that $w_k \in L^2$. Under the effect of $w_k \in L^p$, we now state the definition of the stability of the system considered in \eqref{eqn_system_dynamics}.}

\textcolor{black}{Now, let system \eqref{eqn_system_dynamics} be ISSp, which expresses -- loosely speaking -- probabilistic convergence of the system states to a neighborhood of the origin.
Then, it is intuitive that the system will visit and return to any such neighborhood. 
This property is called \emph{recurrence} and formally defined as follows.}

\begin{definition}(From \cite{Culbertson2023})\label{def:recurrent}
    For a bounded set $\mathcal{S} \in \mathcal{X}$, the hitting time is defined as $\tau_{\mathcal{S}}(x) := \inf \{ k \in \bbn \mid x_k \in \mathcal{S}, x_0 = x \}$. The set $\mathcal{S}$ is recurrent if $\forall x \in \mathcal{X}, \mathbb{P} [ \tau_{\mathcal{S}}(x) < \infty ] = 1$.
\end{definition}

\textcolor{black}{For a proof of the connection between ISSp and recurrence, we refer the reader to \cite{Culbertson2023}.}
}

\end{document}